\documentclass[conference,final]{IEEEtran}
\IEEEoverridecommandlockouts

\usepackage[T1]{fontenc}
\usepackage[utf8]{inputenc}

\usepackage{breakurl}
\usepackage{eucal}

\usepackage{centernot}
\usepackage{mdframed}

\usepackage{amsmath}
\usepackage{amssymb}
\usepackage{amsthm}
\usepackage{paralist}
\usepackage{mathtools}
\usepackage{url}
\usepackage{extarrows}

\usepackage{etoolbox}
\usepackage{xspace} 

\usepackage{stmaryrd} 

\usepackage[inline,shortlabels]{enumitem} 
\newlist{inlinelist}{enumerate*}{1}
\setlist*[inlinelist,1]{%
  label=(\roman*),
}

\usepackage{xifthen}  
\usepackage{ifthen}

\usepackage{nicefrac}

\usepackage[usenames,dvipsnames]{xcolor}

\usepackage{tikz} 
\usepackage{float}

\usepackage{stackengine} 

\usepackage{pifont}

\usepackage{hyperref}
\usepackage{cleveref}

\usepackage[nomargin,inline,index]{fixme} 
\fxusetheme{color}
\definecolor{keywordcolor}{rgb}{0.7, 0.1, 0.1}   
\definecolor{tacticcolor}{rgb}{0.0, 0.1, 0.6}    
\definecolor{commentcolor}{rgb}{0.4, 0.4, 0.4}   
\definecolor{symbolcolor}{rgb}{0.0, 0.1, 0.6}    
\definecolor{sortcolor}{rgb}{0.1, 0.5, 0.1}      
\definecolor{attributecolor}{rgb}{0.7, 0.1, 0.1} 

\definecolor{Black}{HTML}{000000}
\definecolor{Gray}{HTML}{808080}
\definecolor{Magenta}{HTML}{FF00FF}
\definecolor{RubineRed}{HTML}{ED017D}
\definecolor{ForestGreen}{HTML}{028A0F}
\definecolor{MidnightBlue}{HTML}{006795}
\definecolor{Plum}{HTML}{92268F}

\FXRegisterAuthor{bart}{anbart}{\color{magenta} {\underline{bart}}}
\FXRegisterAuthor{zun}{anzun}{\color{ForestGreen} {\underline{zun}}}
\FXRegisterAuthor{ric}{anric}{\color{red} {\underline{ric}}}

\hypersetup{
  breaklinks   = true,
  colorlinks   = true, 
  urlcolor     = blue, 
  linkcolor    = blue, 
  citecolor    = red   
}
\hypersetup{final} 

\usepackage{listings}
\definecolor{listingBG}{HTML}{FFFFCB}%
\definecolor{listingFrame}{HTML}{BBBB98}%
\definecolor{listingLineno}{rgb}{0.5,0.5,1.0}%

\definecolor{LightGrey}{rgb}{0.975,0.975,0.975}

\lstdefinelanguage{solidity}{
, basicstyle=\ttfamily\linespread{1.15}\footnotesize\lst@ifdisplaystyle\footnotesize\fi
, commentstyle=\color{Gray}
, morecomment=[l]{//}
, morecomment=[s]{/*}{*/}
, escapechar=\$
, classoffset=0,
, keywordstyle=\color{NavyBlue}\bfseries
, morekeywords={assert,require,if,then,else,for,break,call,delegatecall,transfer,transferFrom,send,case, catch,continue,do,while,emit, new, return, revert, selfdestruct, try, with, throw, switch, suicide}
, classoffset=1
, keywordstyle=\color{YellowGreen}\bfseries
, morekeywords={external, implements, import, interface, internal, library, payable, pragma, private, protected, public, pure, returns, super, using, view}
, classoffset=2
, keywordstyle=\color{blue}
, morekeywords={function, constructor, contract, constant, struct, address, bool, byte, bytes, bytes1, bytes2, bytes3, bytes4, bytes5, bytes6, bytes7, bytes8, bytes9, bytes10, bytes11, bytes12, bytes13, bytes14, bytes15, bytes16, bytes17, bytes18, bytes19, bytes20, bytes21, bytes22, bytes23, bytes24, bytes25, bytes26, bytes27, bytes28, bytes29, bytes30, bytes31, bytes32, enum, int, int8, int16, int24, int32, int40, int48, int56, int64, int72, int80, int88, int96, int104, int112, int120, int128, int136, int144, int152, int160, int168, int176, int184, int192, int200, int208, int216, int224, int232, int240, int248, int256, mapping, string, uint, uint8, uint16, uint24, uint32, uint40, uint48, uint56, uint64, uint72, uint80, uint88, uint96, uint104, uint112, uint120, uint128, uint136, uint144, uint152, uint160, uint168, uint176, uint184, uint192, uint200, uint208, uint216, uint224, uint232, uint240, uint248, uint256, var, void, ether, finney, szabo, wei, days, hours, minutes, seconds, weeks, years}
, classoffset=3
, keywordstyle=\color{Plum}\bfseries
, morekeywords={balance, block, blockhash, instanceof, coinbase, difficulty, gaslimit, number, timestamp, msg, data, gas, sender, value, sig, value, now, tx, gasprice, origin}
}

\newcommand{\ifempty}[3]{%
  \ifthenelse{\isempty{#1}}{#2}{#3}%
}

\newcommand{\ifdots}[3]{%
  \ifthenelse{\equal{#1}{...}}{#2}{#3}%
}

\newcommand{\hidden}[1]{}

\newcommand{\keyterm}[1]{\textbf{\emph{#1}}}%

\makeatletter
\newcommand*{\itemequation}[3][]{%
  \item
  \begingroup
    \refstepcounter{equation}%
    \ifx\\#1\\%
    \else  
      \label{#1}%
    \fi
    \sbox0{#2}%
    \sbox2{$\displaystyle#3\m@th$}%
    \sbox4{\@eqnnum}%
    \dimen@=.5\dimexpr\linewidth-\wd2\relax
    \ifcase
        \ifdim\wd0>\dimen@
          \z@
        \else
          \ifdim\wd4>\dimen@
            \z@
          \else 
            \@ne
          \fi 
        \fi
      \@latex@warning{Equation is too large}%
    \fi
    \noindent   
    \rlap{\copy0}%
    \rlap{\hbox to \linewidth{\hfill\copy2\hfill}}%
    \hbox to \linewidth{\hfill\copy4}%
    \hspace{0pt}
  \endgroup
  \ignorespaces 
}
\makeatother

\newcommand{\Real}[1]{\mathrm{Real}}

\newcommand{\codefont}{\fontsize{9}{9}\selectfont}

\newcommand{\contract}[1]{{\tt\codefont{\txColor{#1}}}}
\newcommand{\txcode}[1]{{\text{\tt\codefont{\txColor{#1}}}}}

\newcommand{\eg}{e.g.\@\xspace}
\newcommand{\ie}{i.e.\@\xspace}
\newcommand{\wrt}{w.r.t.\@\xspace}

\def\negcaptionspace{\vspace{-10pt}}

 \theoremstyle{plain}
 \newtheorem{thm}{Theorem}
 
 \newtheorem{prop}[thm]{Proposition}
 
 \theoremstyle{definition}

  {}

\def\pmvColor{\color{black}}
\newcommand{\pmvFmt}[1]{{\pmvColor{\tt #1}}}

\newcommand{\pmv}[2][]{\pmvFmt{#2}_{\pmvColor{#1}}\xspace}
\newcommand{\pmvA}[1][]{\pmv[{#1}]{A}} 

\newcommand{\Adv}{\pmv{Adv}} 

\def\txColor{\color{black}}

\DeclareMathAlphabet{\mathbfsf}{\encodingdefault}{\sfdefault}{bx}{n}

\newcommand{\waltok}[2]{#1\!:\!#2}
\newcommand{\walenum}[1]{[#1]}
\newcommand{\walpmv}[2]{{#1}\walenum{#2}}

\newcommand{\waldistrarrow}[1]{\approx_{\$}}

\definecolor{LightGrey}{rgb}{0.95,0.95,0.95}
\definecolor{keyword}{HTML}{7F0055}

\def\tokColor{\color{Magenta}}
\newcommand{\tokFmt}[1]{{\tokColor{\tt #1}}}
\newcommand{\tok}[2][]{\tokFmt{#2}_{\tokColor{#1}}\xspace}

\newcommand{\tokT}[1][]{\tok[{#1}]{t}}    

\newcommand{\priceT}[1]{\texttt{price}({#1})}

\newlength\replength
\newcommand\repfrac{.1}

\newcommand\rulewidth{.6pt}
\newcommand\tdashfill[1][\repfrac]{\cleaders\hbox to \replength{%
  \smash{\rule[\arraystretch\ht\strutbox]{\repfrac\replength}{\rulewidth}}}\hfill}

\newcommand\tdotfill[1][\repfrac]{\cleaders\hbox to \replength{%
  \smash{\raisebox{\arraystretch\dimexpr\ht\strutbox-.1ex\relax}{.}}}\hfill}

\newcommand{\contrAdvC}[2]{\mathcal{C}} 

\def\sysColor{\color{Black}}

\newcommand{\sysFmt}[1]{{\sysColor{#1}}}
\newcommand{\sysS}[1][]{\mathord{\sysFmt{\sigma}_{\sysColor{#1}}}}

\newcommand{\solcode}[1]{\lstinline[language=solidity,basicstyle={\ttfamily\small}]{#1}}

\newtoggle{anonymous}

\begin{document}

\title{Certifying optimal MEV strategies with Lean}

\iftoggle{anonymous}
{
\author{\IEEEauthorblockN{Anonymous}
\IEEEauthorblockA{\textit{University of Nowhere}\\
Nowhere \\
noemail@nowhere.edu}
}
}
{
\author{\IEEEauthorblockN{Massimo Bartoletti}
\IEEEauthorblockA{\textit{University of Cagliari}\\
Cagliari, Italy \\
bart@unica.it}
\and
\IEEEauthorblockN{Riccardo Marchesin}
\IEEEauthorblockA{\textit{University of Trento}\\
Trento, Italy \\
riccardo.marchesin@unitn.it}

\and
\IEEEauthorblockN{Roberto Zunino}
\IEEEauthorblockA{\textit{University of Trento}\\
Trento, Italy \\
roberto.zunino@unitn.it}
}
} 

\maketitle

\begin{abstract}
Maximal Extractable Value (MEV) refers to a class of attacks to decentralized applications where the adversary profits by manipulating the ordering, inclusion, or exclusion of transactions in a blockchain.
Decentralized Finance (DeFi) protocols are a primary target of these attacks, as their logic depends critically on transaction sequencing.
To date, MEV attacks have already extracted billions of dollars in value, underscoring their systemic impact on blockchain security.
Verifying the absence of MEV attacks requires determining suitable upper bounds, \ie proving that no adversarial strategy can extract more value (if any) than expected by protocol designers.
This problem is notoriously difficult: the space of adversarial strategies is extremely vast, making empirical studies and pen-and-paper reasoning insufficiently rigorous.
In this paper, we present the first mechanized formalization of MEV in the Lean theorem prover. 
We introduce a methodology to construct machine-checked proofs of MEV bounds, providing correctness guarantees beyond what is possible with existing techniques.
To demonstrate the generality of our approach, we model and analyse the MEV of two paradigmatic DeFi protocols. 
Notably, we develop the first machine-checked proof of the optimality of sandwich attacks in Automated Market Makers, a fundamental DeFi primitive.
\end{abstract}

\begin{IEEEkeywords}
smart contracts, MEV, decentralized finance, interactive theorem proving, Lean 4
\end{IEEEkeywords}

\section{Introduction}
\label{sec:intro}

Public permissionless blockchains such as Ethereum currently handle billions of dollars in crypto-assets, controlled by smart contracts that implement increasingly complex financial applications.
In most cases, the underlying protocols of these blockchains do not enforce transaction order fairness, instead delegating the sequencing of user transactions to miners or validators.
This leaves decentralized applications vulnerable to Maximal Extractable Value (MEV) attacks, where adversaries manipulate transaction sequencing for profit.
While MEV extraction may have some beneficial effects --- such as reducing transaction fees~\cite{Weintraub22flashbot} --- its overall impact is  detrimental on the affected blockchains, undermining decentralization, transparency, and exacerbating network congestion~\cite{Torres21frontrunner,Qin21quantifying}.
%


MEV can be approached from two different perspectives, depending on whether one  plays the role of attacker or defender.
For an attacker, the fact that the value extracted is \emph{maximal} is not really relevant. 
What truly matters is having an efficient algorithm to bundle one's own transactions with those of other users in a way that guarantees profit.
For that purpose, the attacker can exploit known heuristics targeting specific contracts~\cite{Qin21fc,Daian20flash,Li23ccs}, or devise adaptive techniques that can potentially extract value from arbitrary contracts~\cite{BabelJ0KKJ23ccs}.
In both cases, the adversary wins if the value extracted exceeds the value paid to mount the attack.

Playing the role of defender is substantially harder: to guarantee that the value extractable from a contract is bounded by a given threshold $v$, one must ensure that no adversarial strategy --- among an \emph{infinite} set of possible strategies --- can extract more than $v$.
More abstractly, let $\mathrm{EV}(\sysS,v)$ be a predicate stating that in system state $\sysS$ the extractable value is bounded from above by $v$. 
Then, the adversary's task reduces to \emph{falsification}: finding a counterexample to $\mathrm{EV}(\sysS,v)$, by exhibiting a strategy that extracts some $v' > v$.
Instead, the defender task amounts to \emph{verification}: constructing a proof that $\mathrm{EV}(\sysS,v)$ indeed holds.  

From this perspective, establishing MEV requires the same fundamental ingredients as program verification, namely:
\begin{enumerate}

\item \label{item:intro:model} 
a formal model of the system under analysis, \ie, smart contracts executed on blockchains;

\item \label{item:intro:MEV}
a precise formalization of the property of interest, \ie, $\mathrm{EV}(\sysS,v)$;

\item \label{item:intro:proof}
a proof technique to determine whether a system satisfies or not the property. 

\end{enumerate}

While the literature proposes several formal models of  contracts at different levels of abstraction --- ranging from the low-level Ethereum Virtual Machine~\cite{Hildenbrandt18csf,Grishchenko18post,Cassez23fm} to the high-level  contract language Solidity~\cite{Crafa19fc,Jiao20sp,Marmsoler25fac} ---
the existing MEV formalizations and the associated proof techniques are not fully adequate for the purpose.
The problem is twofold:
\begin{itemize}

\item First, most existing definitions~\cite{Babel23clockwork,Salles21formalization,Mazorra22price} lack the precision needed to establish MEV, as they often omit key aspects of adversarial capabilities such as eavesdropping the transaction mempool~\cite{BZ25fc}.
We emphasize that, from a defender's perspective, the accuracy of a MEV formalization is crucial to ensure that no class of attacks is overlooked.  

\item Second, as noted above, establishing precise MEV bounds requires considering \emph{all} sequences of transactions that an adversary can construct using their private knowledge and the content of the transaction mempool.
Although in many cases this infinite set of adversarial strategies can be partitioned into a finite number of equivalence classes, the resulting combinatorial explosion of cases and subcases quickly becomes unmanageable. This explosion occurs even for relatively simple contracts, where just a handful of system variables already gives rise to a vast and intricate strategy space, well beyond the reach of reliable pen-and-paper proofs.

\end{itemize}

To address these challenges, it is necessary to devise a MEV formalization that is 
\begin{inlinelist}
\item precise enough to capture all possible adversarial strategies,
\item flexible enough to accommodate a wide range of use cases, and 
\item amenable to rigorous, machine-verified proofs that go beyond manual analysis.
\end{inlinelist}
Proof assistants provide a natural framework to meet these requirements, 
as they enable precise formalizations of software systems and support the construction of machine-verified proofs whose correctness is guaranteed beyond any reasonable doubt. 
In the context of MEV, they offer the potential to reason about adversarial strategies, automate parts of the verification process, and ensure a level of rigour that manual analysis cannot achieve.
However, despite its central role in the security of decentralized applications, no mechanized formalization of MEV has been developed so far. 
Existing studies rely on informal arguments, which are insufficient to provide guarantees of correctness. This leaves a critical gap between the theoretical understanding of MEV and the practical need for security of decentralized applications.



\paragraph*{Contributions}
This paper addresses the previous research questions by providing the following key contributions:
\begin{enumerate}

\item we provide the first fully mechanised formalization of MEV
in a proof assistant (\S\ref{sec:model}). To this purpose, we adopt Lean~4~\cite{deMoura21cade}, an open-source theorem prover and programming language that has been successfully applied to construct and verify large-scale proofs~\cite{mathlib}. Its extensive mathematical library makes Lean 4 particularly well suited for reasoning about DeFi contracts, which often involve complex mathematical manipulations.

\item we devise a new proof technique for establishing MEV. 
Roughly, given a system state $\sysS$ where MEV is to be estimated, our proof technique requires the defender to provide a ``guess'' function mapping each state to a candidate MEV amount.
The defender must then prove that
\begin{inlinelist}
\item the guess for the initial state is an under-approximation of MEV, and 
\item any adversarial move affects the guess 
by an amount which is bounded by the gain of the move.
\end{inlinelist}
We establish that our proof technique is \emph{sound} --- \ie, when such a guess function exists, it actually gives the MEV ---
and \emph{complete} --- \ie, when the MEV exists, a guess function exists.

\item we apply our proof technique to 
two paradigmatic case studies.
The first is a gambling game in which players can win a prize by depositing suitable amounts of tokens.
Despite its apparent simplicity, this case study highlights how different adversarial strategies can emerge depending on the contents of the transaction mempool (\S\ref{sec:coinpusher}). 
The second contract is an Automated Market Maker (AMM), one of the cornerstones of Decentralized Finance~\cite{Angeris20aft,Werner22aft}. 
Although the MEV of AMMs has been studied before~\cite{Zhou21high,BCL22fc,Xu23csur,Kulkarni23wine}, our work provides the first mechanized proof that so-called \emph{sandwich attacks}  extract the maximal possible value (\S\ref{sec:amm}).

\end{enumerate}

Our Lean implementation, including all the proofs and case studies are available online in \iftoggle{anonymous}{an \href{https://anonymous.4open.science/r/MEV-formal-EF68}{anonymous public repository}~\cite{mev-lean-github-anonymous}}{a \href{https://github.com/r-marche/MEV-formal}{public repository}~\cite{mev-lean-github}}
consisting of ${\sim}$7000 lines of code.  



\section{Background on MEV}
\label{sec:background}

At an abstract level, we can see a blockchain as a transition system, where states represent both users' wallets (\ie, their token holdings) and the states of deployed contracts (including their balances).
State transitions are triggered by \emph{transactions} sent by users:
a transaction may affect the sender's wallet, the state of the called contracts, and the wallets receiving tokens from those contracts (if any).

As a concrete example, consider an \emph{Airdrop} contract that allows any user to withdraw tokens from its balance. 
We specify this contract in Solidity in~\Cref{fig:airdrop}.
Deploying the contract requires the sender to transfer any positive amount  of tokens from its wallet to the contract.
In this case, the tokens correspond to the blockchain native cryptocurrency (\eg, ETH on Ethereum), and they are quantified by the expression \solcode{msg.value}.
Besides the constructor, the contract only features another function, \solcode{drop}, which allows any user (identified by \solcode{msg.sender}) to withdraw an arbitrary fraction of the contract balance. 
The contract state only consists of its balance, which is implicitly updated when receiving tokens (in the \solcode{constructor}) or sending them (in the \solcode{transfer} command). 

MEV quantifies the maximal gain that an adversary can obtain by exploiting their power to reorder, drop, or insert transactions in the blockchain.
To this end, the adversary can play both pending user transactions in the public mempool, and may also inject its own crafted transactions, potentially leveraging knowledge of the mempool contents.

In our working example, assume a system state $\sysS$ consisting of the sole Airdrop contract with a balance of $n > 0$ tokens, which for simplicity we assume to have unitary price. 
Any adversary can simply fire a \solcode{drop(n)} transaction to empty the contract balance: therefore, the MEV in $\sysS$ is exactly $n$. 
Note that the Airdrop's MEV is not necessarily ``bad MEV'', since the extraction of value is aligned with the intended functionality of the contract.
However, an adversary with transaction sequencing powers can always front-run \solcode{drop} transactions of honest users, depriving them of any gain.
In the rest of the section, we will examine cases of ``bad MEV'' where the adversary causes a loss to honest users.

\begin{figure}
\begin{lstlisting}[language=solidity,frame=single]
contract AirDrop {
  constructor() payable { // receive tokens from the sender
      require (msg.value > 0); 
  }
  function drop(uint v) public {
      address payable rcv = payable(msg.sender);
      rcv.transfer(v); // transfer tokens to the sender
  }
}\end{lstlisting}
\negcaptionspace
\caption{An Airdrop contract in Solidity.}
\label{fig:airdrop}
\end{figure}

In the previous case, the adversary does not need to exploit the mempool, but just to front-run other users' transactions (this is called \emph{displacement attack} in~\cite{Eskandari19sok}).
More sophisticated forms of MEV arise when the adversary exploits pending transactions in the mempool, \eg by constructing bundles that combine users' transactions with adversarial ones.
We will see in the rest of the section how these strategies make the estimation of MEV increasingly more complex.

\begin{figure}
\begin{lstlisting}[language=solidity,frame=single]
contract CoinPusher {
  function push() public payable {
    if (address(this).balance >= 100) {
      address payable rcv = payable(msg.sender);
      rcv.transfer(address(this).balance);
    }
  }
}\end{lstlisting}
\negcaptionspace
\caption{A CoinPusher contract in Solidity.}
\label{fig:coinpusher}
\end{figure}

\subsection{The CoinPusher contract}
\label{sec:background:coinpusher}

We now consider a case where extracting MEV
requires the adversary to leverage transactions pending in the mempool.
The \emph{CoinPusher} contract transfers its entire balance to any user whose deposit 
causes the balance to exceed $100$ tokens (~\Cref{fig:coinpusher}).
The contract has a single function \solcode{push}, which receives from the sender any amount \solcode{msg.value} of tokens (this transfer happens implicitly along with the call). 
If the incoming tokens make the contract balance exceed 100 units, the entire balance (including the incoming tokens) is immediately sent to the caller through the \solcode{transfer} command. 

Let $\sysS$ be a system state where the contract holds $0$ tokens
and a user $\pmvA$ holds $1$ token. 
We assume the adversary $\Adv$ to be \emph{wealthy}, \ie endowed with enough tokens to mount any feasible attack
(in practice, $\Adv$ can also obtain such tokens as a very short-term loan, even if that might cost of some additional fees).
If the mempool in $\sysS$ is empty, then $\Adv$ cannot extract any value, hence ${\rm MEV}(\sysS) = 0$.
Now suppose that the mempool contains a transaction sent by $\pmvA$ that calls \txcode{push} while transferring 1 token ---
written \mbox{$\pmvA:\txcode{push(value=1)}$}. 
In this case, $\Adv$ can atomically execute the transaction bundle:
\[
\pmvA:\txcode{push(value = 1)} \quad \Adv:\txcode{push(value = 99)}
\]
which yields a gain of $1$ to $\Adv$. 
Since no strategy achieves a higher gain, we conclude that ${\rm MEV}(\sysS) = 1$.
Devising an optimal strategy for an arbitrary contract when the mempool contains many transactions is hard, also due to the combinatorial explosion of possible interleavings. 
In the CoinPusher case, the adversary's optimal strategy is to include the pending mempool transactions while interleaving its own \solcode{push} calls. 
In the idealized case where there are no transaction fees, this can yield a gain equal to the contract balance in $\sysS$ plus the total value of those mempool transactions with value $<100$. 

\begin{figure}[t]
\begin{lstlisting}[language=solidity,frame=single]
contract AMM {
  uint r0, r1; // reserves of token types t0, t1
  
  constructor(uint n0, uint n1) {
    t0.transferFrom(msg.sender, address(this), n0);
    t1.transferFrom(msg.sender, address(this), n1);
    r0 = n0; r1 = n1;
  }
  function swap0(uint n0, uint x1_min) {
    t0.transferFrom(msg.sender, address(this), n0);
    x1 = (n0 * r1) / (r0 + n0); // compute output tokens
    require x1>=x1_min && x1<r1;
    t1.transfer(msg.sender, x1); 
    r0 = r0 + n0; r1 = r1 - x1;
  }
  function swap1(uint n1, uint x0_min) {
    // symmetric to swap0
  }
}
\end{lstlisting}
\negcaptionspace
\caption{A constant-product AMM contract in Solidity (simplified).}
\label{fig:amm}
\end{figure}

\subsection{Automated Market Makers}
\label{sec:background:amm}

We now consider an Automated Market Maker (AMM), an archetypal decentralized finance (DeFi) primitive that enables users to swap between two token types according to an algorithmically defined exchange rate~\cite{Angeris20aft,Angeris21analysis,BCL22amm}.

For illustration, we present in~\Cref{fig:amm} a simplified Solidity code of the AMM contract (an actual implementation is substantially more complex, \eg because it has to deal with rounding of integer operations).
The \solcode{constructor} initializes the reserves of tokens $\tokT[0]$ and $\tokT[1]$, which are transferred from the sender's wallet to the contract.
The function \solcode{swap0} allows anyone to send \solcode{n0} units of $\tokT[0]$ and receive at least \solcode{x1_min} units of $\tokT[1]$.
The function \solcode{swap1} is symmetric: it takes $\tokT[1]$ as input and outputs $\tokT[0]$.
The exchange rate follows the mechanism of Uniswap v2~\cite{uniswapimpl}, which maintains the product of reserves (\ie, \mbox{\solcode{r0} $\cdot$ \solcode{r1}}) constant across swaps. 
This guarantees that the marginal exchange rate (\ie, the one applied for infinitesimal swaps) coincides with the ratio between the reserves.

When external prices are not aligned with the marginal exchange rate --- \ie, when 
\mbox{\solcode{r0}$\,\cdot\, \priceT{\tokT[0]} \neq$ \solcode{r1}$\,\cdot\,  \priceT{\tokT[1]}$} ---
a strategy to extract value is to perform an \emph{arbitrage}, \ie fire a swap that realigns the reserves with external prices.
For example, let $\sysS$ be a state where $\priceT{\tokT[0]} = 4$,  $\priceT{\tokT[1]} = 9$, and the AMM reserves are of $6$ units of each token,
written: 
\[
\sysS =
\walpmv{\contract{AMM}}{\waltok{6}{\tokT[0]},\waltok{6}{\tokT[1]}}
\mid
\walpmv{\Adv}{\waltok{n_0}{\tokT[0]},\waltok{n_1}{\tokT[1]}}  
\] 
If the adversary calls \mbox{\txcode{swap0}(3,0)}, the state becomes:  
\[
\walpmv{\contract{AMM}}{\waltok{9}{\tokT[0]},\waltok{4}{\tokT[1]}}
\mid
\walpmv{\Adv}{\waltok{n_0-3}{\tokT[0]},\waltok{n_1+2}{\tokT[1]}}  
\] 
The resulting AMM is balanced, and the adversary has paid 3 units of $\tokT[0]$ (with value $3 \cdot 4 = 12$) 
to buy $2$ units of $\tokT[1]$ (with value $2 \cdot 9 = 18$).
So, overall $\Adv$ has gained $18 - 12 = 6$.
One of our contributions (\S\ref{sec:amm:empty-mempool}) is a machine-checked proof that arbitrage is indeed the optimal MEV-extracting strategy when the mempool is empty. 

If the mempool is nonempty, more complex strategies arise, which potentially give a higher MEV.
A typical strategy is the so-called \emph{sandwich attack}~\cite{Zhou21high},
which we illustrate in the same state $\sysS$ as before. 
Assume that a honest user $\pmvA$ has sent an arbitrage transaction
\mbox{$\pmvA:\txcode{swap0}(3,1)$} to the mempool.
Here, $\pmvA$ has set the parameter \solcode{x1_min} to $1$, enforcing a lower bound on the number of $\tokT[1]$ units received from the swap. 
This safeguard is meant to account for the uncertainty of the state in which a transaction will be actually executed --- an inherent problem in account-based blockchains. 
In a sandwich attack, the adversary has access to the mempool
(containing $\pmvA$'s arbitrage), and uses it to construct the following transaction bundle:
\begin{enumerate}
  
\item \mbox{$\Adv:\txcode{swap0}(3,0)$}, an arbitrage transaction performed by the adversary;
  
\item \mbox{$\pmvA:\txcode{swap0}(3,1)$}, the transaction picked from the mempool. 
Since $\pmvA$'s transaction is executed in a state where the AMM is in equilibrium, $\pmvA$ will not receive the 2 units they would have expected in $\sysS$: rather, they receive the minimum of $\tokT[1]$ units admitted by the constraint \solcode{x1_min}, \ie only 1 unit.
This causes $\pmvA$ to have a negative gain;

\item \mbox{$\Adv:\txcode{swap1}(1,3)$}.
The adversary closes the sandwich with another arbitrage transaction,
after which the AMM reaches again the equilibrium.
      
\end{enumerate}

Overall, executing this transaction bundle in $\sysS$ leads to:
\begin{align*}
& \walpmv{\contract{AMM}}{\waltok{6}{\tokT[0]},\waltok{6}{\tokT[1]}}
\mid
\walpmv{\Adv}{\waltok{n_0}{\tokT[0]},\waltok{n_1}{\tokT[1]}}
\mid 
\cdots
\\
\xrightarrow{(1)} \;
& \walpmv{\contract{AMM}}{\waltok{9}{\tokT[0]},\waltok{4}{\tokT[1]}}
\mid
\walpmv{\Adv}{\waltok{n_0-3}{\tokT[0]},\waltok{n_1+2}{\tokT[1]}}
\mid 
\cdots
\\
\xrightarrow{(2)} \;
& \walpmv{\contract{AMM}}{\waltok{12}{\tokT[0]},\waltok{3}{\tokT[1]}}
\mid
\walpmv{\Adv}{\waltok{n_0-3}{\tokT[0]},\waltok{n_1+2}{\tokT[1]}}
\mid 
\cdots
\\
\xrightarrow{(3)} \;
& \walpmv{\contract{AMM}}{\waltok{9}{\tokT[0]},\waltok{4}{\tokT[1]}}
\mid
\walpmv{\Adv}{\waltok{n_0}{\tokT[0]},\waltok{n_1+1}{\tokT[1]}}
\mid 
\cdots
\end{align*}

This gives the adversary a gain of $1 \cdot \priceT{\tokT[1]} = 9$, which is greater than the gain obtainable without exploiting the mempool, and it actually turns out to be the MEV in $\sysS$.
More in general, devising the optimal MEV-extracting strategy depends on a multitude of factors: the AMM reserves, the external token prices, the direction of the swap in the mempool (\ie, \solcode{swap0} \emph{vs.} \solcode{swap1}), as well as the swap amount and lower bound. 
The resulting combinatorial explosion of cases and subcases makes manual reasoning about MEV extremely error-prone.
Our machine-checked proof in Lean (\S\ref{sec:amm:nonempty-mempool}) addresses this complexity, by establishing the MEV when the adversary can exploit a transaction from the mempool. 
\section{System model}
\label{sec:model}

In this section, we introduce our Lean formalization of smart contracts executed on blockchains (\S\ref{sec:model:systate}). 
We illustrate it by formalizing the \emph{Airdrop} contract presented in~\S\ref{sec:background} (\S\ref{sec:model:airdrop}).

\subsection{System state}
\label{sec:model:systate}

A \emph{wallet} is a container of tokens, possibly of different types.
We model wallets as functions from token types to real-valued amounts.
We explicitly separate the adversary from honest users, reflecting the different assumptions we make about their capabilities.
In particular, we assume our adversary to be \emph{wealthy}, \ie, able to spend arbitrarily large amounts of tokens when mounting an attack. Honest participants, instead, own a limited amount of tokens, as in the real world.
We formalize these assumptions by defining two distinct wallet types: 
\lstinline{Wallet}, holding a non-negative amount of tokens for honest participants and contracts, and \lstinline{WalletAdv}, holding arbitrary token amounts (possibly negative) for the adversary.
We parametrize these types by \lstinline{Token}, a type representing the token types:

\begin{lstlisting}
def Wallet    (Token : Type) : Type := Token → ℝ≥0
def WalletAdv (Token : Type) : Type := Token → ℝ
\end{lstlisting}

We parameterise our model with a type \lstinline{State}, which represents the state of an arbitrary contract running in an honest environment.
This parameter will be instantiated when defining specific contracts, \eg in~\S\ref{sec:model:airdrop}, \S\ref{sec:coinpusher} and \S\ref{sec:amm}.
Intuitively, such \lstinline{State} comprises the contract variables, its wallet, the wallets of the honest participants, and the mempool.

We model the interactions between the adversary, the honest participants, and the contract as a state transition system.
Its states are types \lstinline{SysState}, consisting of a \lstinline{State} and the adversary's wallet:
\begin{lstlisting}
structure SysState {Token State : Type} where
  Δ : WalletAdv Token
  s : State
\end{lstlisting}

The rules of the transition system are defined by the structure \lstinline{System} (\Cref{fig:lean:System}), which specifies the types:
\begin{itemize}
    \item \lstinline{Token}, representing the possible token types exchanged within the system;
    \item \lstinline{State}, mentioned above;
    \item \lstinline{Move}, representing the possible adversarial moves, 
    \eg, crafting and executing a transaction, or fetching from the mempool a transaction sent from an honest participant and executing it.
\end{itemize}

A \lstinline{System} must provide a \lstinline{semantics} for adversarial moves (note that honest participant's moves are already taken into account by the transactions in the mempool).
The \lstinline{semantics} is a \emph{partial} function, mapping a \lstinline{SysState} and a \lstinline{Move} to a new \lstinline{SysState}. 
The semantics is undefined (\lstinline{none}) when it is impossible to perform the given move.
In practice, this corresponds to the case where a transaction reverts.

In order to prove general properties on \lstinline{System}s --- \ie, properties that do no depend on the specific instantiation of its semantics ---
we require a few more fields and assumptions.
First, \lstinline{System} must provide a function (\lstinline{honTokens}) mapping each \lstinline{State} to the cumulative amount of tokens owned by the  contract and honest participants
(hereafter, referred to as the ``honest tokens''). 
%
Second, we postulate that the \lstinline{semantics} preserves tokens: we model this through the assumption \lstinline{preserveTokens}, which forbids the minting and the burning of tokens.
Finally, a \lstinline{System} must associate each token type with a \emph{price}:
formally, this is modelled as a function \lstinline{tokenValue} that
maps any wallet to a real number, denoting the cumulative price of all the tokens in the wallet.
We require this value function to be non-negative and additive.

\begin{figure}[t]
\begin{lstlisting}[frame=single]
structure System where
  Token : Type
  State : Type
  Move  : Type
  semantics : @SysState Token State → Move →
    Option (@SysState Token State)
  honTokens : State → Wallet Token
  preserveTokens : ∀ σ m,
    match semantics σ m with
    | .none    => True
    | .some σ' => ∀ τ, honTokens σ.s τ + σ.Δ τ
                     = honTokens σ'.s τ + σ'.Δ τ               
  tokenValue : WalletAdv Token → ℝ
  tokenValue_nonneg : ∀ f, f ≥ 0 → tokenValue f ≥ 0
  tokenValue_additive : ∀ f g,
    tokenValue (f + g) = tokenValue f + tokenValue g
\end{lstlisting}
\negcaptionspace
\caption{\lstinline{System} model in Lean.}
\label{fig:lean:System}
\end{figure}

Given \lstinline{sys : System}, we can then prove a few basic properties. For instance, we establish that the value of the empty wallet is zero, and that the value function preserves subtraction and is monotonic on wallets.

\vbox{
\begin{lstlisting}
theorem tokenValue_zero : sys.tokenValue 0 = 0
theorem tokenValue_sub (f g : WalletAdv _) : sys.tokenValue (f - g) 
  = sys.tokenValue f - sys.tokenValue g
theorem tokenValue_monotonic (f g : WalletAdv _) :
  f ≤ g → sys.tokenValue f ≤ sys.tokenValue g
\end{lstlisting}
}

For readability, hereafter we abbreviate \lstinline{SysState}, instantiated with \lstinline{sys.Token} and \lstinline{sys.State}, as \lstinline{sys.sysState}:
\begin{lstlisting}
abbrev System.sysState : Type :=
  @SysState sys.Token sys.State
\end{lstlisting}

\subsection{Example: formalization of the Airdrop contract}
\label{sec:model:airdrop}

We exemplify our Lean formalization to model the Airdrop contract introduced in~\Cref{sec:background}.
We start by providing some general definitions regarding the participants, the exchanged tokens, and the the transactions.
This skeleton will be reused also for the other use cases.

We define a type \lstinline{Participant} to describe who can interact with the contract. The type comprises infinitely many honest participants and an adversary.

\vbox{%
\begin{lstlisting}
inductive Participant
| Hon : String → Participant
| Adv : Participant
\end{lstlisting}
}

The type \lstinline{Token} denotes the token types exchangeable through the contract. 
Our Airdrop uses a single token type:

\vbox{
\begin{lstlisting}
inductive Token
| τ₀ : Token
\end{lstlisting}
}

Coherently with the Solidity code in~\Cref{fig:airdrop}, 
the Airdrop contract features a single function \lstinline{drop}, which allows any participant \lstinline{P} to 
withdraw any amount \lstinline{v} $> 0$ of tokens from the contract,
failing if \lstinline{v} exceeds the contract balance.
Correspondingly, transactions will take the following form:

\vbox{
\begin{lstlisting}
inductive Tx
| drop (P : Participant) (v : ℝ+) : Tx
\end{lstlisting}
}

The argument \lstinline{P} in a \lstinline{drop} transaction implicitly denotes that the transaction is signed by \lstinline{P}.
We assume that transactions are not \emph{malleable}, \ie the adversary cannot alter the parameter \lstinline{v} without invalidating the signature.

The \lstinline{State} type is a structure that denotes the state of the  blockchain without the adversary. More specifically, the field \lstinline{bal} represents the contract's balance, while the field \lstinline{wal} denotes the aggregated wallets of all honest participants.
Since MEV analysis does not require distinguishing between individual honest users, their holdings are abstracted into this single cumulative wallet.
%
The field \lstinline{mempool} represents the transactions previously broadcast by honest participants but not yet finalised on-chain.
Formally, we represent the \lstinline{mempool} as an association list 
mapping transaction identifiers (\lstinline{TxId}) with their associated transactions.

\vbox{
\begin{lstlisting}
structure State where
  bal : Token → ℝ≥0
  wal : Wallet Token
  mempool : AssocList TxId Tx
\end{lstlisting}
}

We abbreviate the associated system state as \lstinline{ADState}:
\begin{lstlisting}
abbrev ADState := @SysState State Token
\end{lstlisting}

We now turn to defining the type of adversarial moves.
To this aim, we start by identifying the subset of transactions that can be crafted by the adversary alone. In our scenario, the adversary can only craft \lstinline{drop} transactions signed by \lstinline{Adv} itself.
This is formalized by property \lstinline{advTx}.

\vbox{
\begin{lstlisting}
inductive advTx : Tx → Prop
| drop {v} : advTx (.drop .Adv v)
\end{lstlisting}
}

The \lstinline{Move} type captures all possible adversarial actions, namely:
\begin{inlinelist}
\item craft a transaction using its own knowledge (\ie, $\Adv$'s private key) and append it directly to the blockchain (\lstinline{adv}), or 
\item select a transaction from the mempool and include it in the blockchain (\lstinline{mempool}).
\end{inlinelist}

\vbox{
\begin{lstlisting}
inductive Move
| adv : (t : Tx) → advTx t → Move
| mempool : TxId → Move
\end{lstlisting}
}

The semantics of a \lstinline{Move} is to cause an \lstinline{ADState} update, or fail, coherently with \lstinline{System.semantics}.
More precisely, an \lstinline{adv} move transfers the specified amount of tokens to the adversary, removing tokens from the contract balance (\lstinline{bal}) and adding them to the adversary's (\lstinline{Δ}, in the system state), accordingly.
If there are not enough tokens in \lstinline{bal}, the move has no effect.

Instead, a \lstinline{mempool} move looks up the corresponding transaction in the \lstinline{State.mempool} field. 
If there is no such transaction, the move has no effect. 
Otherwise, if we find a corresponding \lstinline{drop P v} transaction, we attempt to execute it, sending \lstinline{v} tokens to \lstinline{P}. 
If there are not enough tokens in \lstinline{bal}, the move has no effect. 
Otherwise, we update the contract balance (\lstinline{bal}) and \lstinline{P}'s wallet.
We distinguish between two cases: if \lstinline{P} is honest, then we update 
\lstinline{wal}, while if it is the adversary we update \lstinline{Δ}.
%
%
If the mempool transaction succeeds, it is removed from the \lstinline{mempool} field so that it is not reused later on.

Using the above semantics, we define the \lstinline{Airdrop} type as a \lstinline{System} modelling the whole contract.
While doing that, we define the \lstinline{Airdrop.honTokens} field, which counts all the circulating tokens not owned by \lstinline{Adv}, by essentially summing \lstinline{bal} and \lstinline{wal}.
We also price the token type $\tau_0$, thereby defining the \lstinline{Airdrop.tokenValue} function which assigns a value to any wallet.
We finally prove the required properties according to the other \lstinline{System} fields, \ie, \lstinline{preserveTokens}, \lstinline{tokenValue_nonneg}, and \lstinline{tokenValue_additive}.

\begin{lstlisting}
def Airdrop : System where …
\end{lstlisting}

Based on this system model, we will analyse the MEV of the Airdrop contract in~\S\ref{sec:mev:airdrop}. 

\section{MEV}
\label{sec:mev}

In this~\namecref{sec:mev} we present our Lean formalization of MEV and a general proof technique for certifying the MEV of contracts.
We start in~\S\ref{sec:mev:gain} by defining the \emph{gain} of the adversary upon performing moves, which is the basis for defining MEV later in~\S\ref{sec:mev:mev}.
In~\S\ref{sec:mev:characterization} we introduce our proof technique in the form of a MEV characterization theorem.
We illustrate such technique in~\S\ref{sec:mev:airdrop}, by applying it to our Airdrop contract.

\subsection{Gain}
\label{sec:mev:gain}

We define the adversarial gain between two system states as the difference of the adversarial wallets. A simple transitivity-like property follows.

\vbox{
\begin{lstlisting}
def gainState (σ σ' : sys.sysState) : ℝ :=
  sys.tokenValue σ'.Δ - sys.tokenValue σ.Δ
  
theorem gainState_trans (σ σ' σ'' : sys.sysState) :
  gainState sys σ σ' + gainState sys σ' σ'' 
  = gainState sys σ σ''
\end{lstlisting}
}

In order to define the gain achieved by the adversary by performing a \emph{list} of moves,
we first generalize the \lstinline{semantics} of adversarial moves. 
The effect of a list of moves is the composition of the individual effect of each move in the list, discarding those that fail.
For brevity, we omit the definition of this function --- hereafter, we will similarly omit long definitions, referring to
\iftoggle{anonymous}{\cite{mev-lean-github-anonymous}}{\cite{mev-lean-github}} for their Lean code.
\begin{lstlisting}
def semMoves (σ : sys.sysState)
  (tr : List sys.Move) : sys.sysState
\end{lstlisting}

The adversarial gain of a list of moves is given by comparing
the initial and final states.
\begin{lstlisting}
def gainMoves 
    (σ : sys.sysState) (tr : List sys.Move) : ℝ := 
  gainState sys σ (semMoves sys σ tr)
\end{lstlisting}

We establish a few basic properties of \lstinline{gainMoves}:
the empty list gives zero gain, and the gain of the first move in a list can be separated from the gain of the rest of the moves.

\vbox{
\begin{lstlisting}
theorem gainMoves_empty (σ : sys.sysState) :
  gainMoves sys s [] = 0
  
theorem gainMoves_step (σ : sys.sysState)
  (m : sys.Move) (ms : List sys.Move) :
    gainMoves sys σ (m :: ms)
    = match sys.semantics σ m with
    | .none    => gainMoves sys σ ms
    | .some σ' => gainState sys σ σ' + gainMoves sys σ' ms
\end{lstlisting}
}
The gain of a list of moves is bounded above by the value of all the circulating honest tokens.
This will be exploited to obtain an upper bound to the extractable value.

\vbox{
\begin{lstlisting}
theorem gainMoves_bound (σ : sys.sysState)
    (tr : List sys.Move) :
  gainMoves sys σ tr ≤ sys.tokenValue (sys.honTokens σ.s)
\end{lstlisting}
}

\subsection{MEV definition}
\label{sec:mev:mev}

A system state \lstinline{σ} has MEV equal to \lstinline{v} when two conditions are met:
\begin{enumerate*}[$(i)$]
    \item there is a list of moves \lstinline{tr} that, when performed in \lstinline{σ}, gives the adversary a gain of \lstinline{v};
    \item any list of moves \lstinline{tr}, when performed in \lstinline{σ}, gives the adversary a gain of at most \lstinline{v}:
\end{enumerate*}

\vbox{
\begin{lstlisting}
structure MEV (σ : sys.sysState) (v : ℝ) : Prop where
  trace_reaches_v :    ∃ tr, gainMoves sys σ tr = v
  other_traces_worse : ∀ tr, gainMoves sys σ tr ≤ v
\end{lstlisting}
}

\zunnote{Forse meglio spostare il MEVsup in limitations}
We remark that MEV is not always guaranteed to exist. 
For instance, if the adversary can extract any real value $< 1$, but they cannot extract $1$, there is no \emph{maximum} value that can be extracted, but only a least upper bound. 
To account for this, we also define the ``least upper bound of the extractable values'' \lstinline{MEVsup}. The definition only requires traces whose gain is arbitrarily close to \lstinline{v}.
We prove that \lstinline{MEVsup} always exists, is unique, and non-negative.
Further, when \lstinline{MEV} exists, it coincides with \lstinline{MEVsup}, so \lstinline{MEV} is unique and non-negative. 

\bartnote{forse sacrificabile la def Lean}

\vbox{
\begin{lstlisting}
structure MEVsup (σ : sys.sysState) (v : ℝ) : Prop where
  traces_approx : ∀ (ε : ℝ+), ∃ tr,
    gainMoves sys σ tr + ε ≥ v
  other_traces_worse : ∀ tr, gainMoves sys σ tr ≤ v
\end{lstlisting}
}


\subsection{Proving MEV}
\label{sec:mev:characterization}

In principle, one could attempt to prove that a system \lstinline{sys} in state \lstinline{σ} has a given MEV equal to \lstinline{v} by directly applying the definition of \lstinline{MEV}.
Doing that would require to 
\begin{enumerate*}
\item find a trace giving $\Adv$ a gain \lstinline{v}, and 
\item show that no other trace can extract more value.
\end{enumerate*}
The second task, in particular, is quite hard, as it requires reasoning on \emph{all} the (infinitely many) adversarial traces.
To address this challenge, we introduce a \emph{MEV characterization theorem} that provides a principled proof technique for establishing that a given value indeed coincides with the MEV.
More precisely, using \lstinline{MEV.characterization} theorem (\Cref{fig:lean:MEV.characterization}) requires following these steps:
\begin{enumerate}

    \item Fix the system state \lstinline{σ} from which the adversary is going to extract MEV.

    \item Specify an \keyterm{invariant} \lstinline{inv} on system states. The invariant must be accompanied with a proof ensuring that
    it holds on \lstinline{σ} (\lstinline{invariant_init}), and
    that it is preserved by any adversarial move (\lstinline{invariant_sound}).

    \item Specify a \keyterm{guess function} \lstinline{MEV_guess} that maps any state satisfying the invariant to a value in \lstinline{ℝ≥0}, which is a candidate MEV for that state.

    \item Find an adversarial trace that extracts from
    \lstinline{σ} exactly \lstinline{MEV_guess σ}. 
    This ensures that, on the system state \lstinline{σ}, the guess function provides a lower bound to the MEV.
    We call this property \keyterm{coherence} (\lstinline{MEV_guess_coherent}).

    \item Prove that the guess function is an upper bound for MEV. 
    More precisely, we must prove that, if \lstinline{σ₀} is any (invariant-abiding) system state, then moving to another state \lstinline{σ₁} cannot induce a gain for the adversary that is greater than the difference between the guesses.
    We call this property \keyterm{soundness} (\lstinline{MEV_guess_sound}):
    \[\small
    \texttt{MEV\_guess}\ \sigma_0
    \geq
    \texttt{gainState}\ \sigma_0\ \sigma_1 + \texttt{MEV\_guess}\ \sigma_1
    \]
    
\end{enumerate}

Once all the steps above are completed, our characterization theorem establishes \lstinline{MEV σ (MEV_guess σ)}, proving that the MEV in \lstinline{σ} is indeed the one given by the guess function.

The choices of the invariant and of the guess function are \emph{crucial}.
Intuitively, the invariant defines an over-approximation of the states that can be reached from \lstinline{σ} by the adversary.
Choosing a suitable invariant is key to simplify the subsequent steps, since they involve properties that are only required to hold on invariant-abiding states.
In our experience, simple invariants are enough to this purpose:
\eg, both in the CoinPusher and in the AMM, it is enough to impose a bound on the size of the mempool. 
The other crucial step is choosing the guess function. Roughly, this corresponds to estimate the value extracted by an adversary following the optimal  strategy.
Fortunately, this estimate must only be provided for the states satisfying the invariant, simplifying this task.  

\begin{figure}[t]
\begin{lstlisting}[frame=single]
theorem MEV.characterization
  (σ : sys.sysState)
  (inv : sys.sysState → Prop)
  (MEV_guess : (σ : sys.sysState) → inv σ → ℝ≥0)
  (invariant_init : inv σ)
  (invariant_sound : 
    ∀ {m : sys.Move} {σ₀ σ₁ : sys.sysState},
    sys.semantics σ₀ m = some σ₁ →
    inv σ₀ → inv σ₁)
  (MEV_guess_coherent : ∃ tr : List sys.Move,
    gainMoves sys σ tr = MEV_guess σ invariant_init)
  (MEV_guess_sound :
    ∀ {m : sys.Move} {σ₀ σ₁ : sys.sysState},
    (hmove : sys.semantics σ₀ m = some σ₁) →
    (σ₀_inv : inv σ₀) →
    let σ₁_inv := invariant_sound hmove σ₀_inv
    gainState sys σ₀ σ₁ + MEV_guess σ₁ σ₁_inv ≤ MEV_guess σ₀ σ₀_inv) :
  MEV sys σ (MEV_guess σ invariant_init)
\end{lstlisting}
\negcaptionspace
\caption{MEV characterization theorem (soundness).}
\label{fig:lean:MEV.characterization}
\end{figure}

Our Lean formalization of the MEV characterization theorem closely follows the previous informal discussion, with 
the only technical difference that one more argument must be passed to the guess function to ensure the invariant holds.

Below, we sketch the proof for \lstinline{MEV.characterization}.
\begin{proof}[Proof (sketch)]
We have to prove that \lstinline{MEV_guess σ} is the MEV in \lstinline{σ}.
The item \lstinline{traces_reaches_v} follows directly from \lstinline{MEV_guess_coherent}.
For \lstinline{other_traces_worse}, consider an arbitrary trace
$m_0,\ldots,m_{n-1}$ of moves starting from $\sigma_0=\sigma$ and going through states $\sigma_1, \sigma_2, \ldots, \sigma_n$.
We have that:

{\small
\begin{align*}
    & \texttt{gainMoves}\ \sigma\ [m_0,\ldots, m_{n-1}]
    \\
    = & \phantom{+} \texttt{gainState}\ \sigma\ \sigma_1 
    \\
    & +\cdots
    \\
    &+ \texttt{gainState}\ \sigma_{n-2}\ \sigma_{n-1}
    \\
    &+ \texttt{gainState}\ \sigma_{n-1}\ \sigma_n
    \\
    \leq & \phantom{+} \texttt{gainState}\ \sigma\ \sigma_1 
    \\
    & +\cdots
    \\
    &+ \texttt{gainState}\ \sigma_{n-2}\ \sigma_{n-1}
    \\
    &+ \texttt{gainState}\ \sigma_{n-1}\ \sigma_n
    \\
    &+ \texttt{MEV\_guess}\ \sigma_n
    && \text{since $\texttt{MEV\_guess}\ \sigma_n \geq 0$}
    \\
    \leq & \phantom{+} \texttt{gainState}\ \sigma\ \sigma_1 
    \\
    & +\cdots
    \\
    &+ \texttt{gainState}\ \sigma_{n-2}\ \sigma_{n-1}
    \\
    &+ \texttt{MEV\_guess}\ \sigma_{n-1}
    && \text{by \lstinline{MEV_guess_sound}}
    \\
    \leq & \phantom{+} \texttt{gainState}\ \sigma\ \sigma_1 
    \\
    & +\cdots
    \\
    &+ \texttt{MEV\_guess}\ \sigma_{n-2}
    && \text{by \lstinline{MEV_guess_sound}}
    \\
    \leq & \phantom{+} \cdots
    \\
    \leq & \phantom{+} \texttt{MEV\_guess}\ \sigma 
    && \text{by \lstinline{MEV_guess_sound}}
\end{align*}
}
%
%
%
%
In the chain of inequalities above we repeatedly apply the \lstinline{MEV_guess_sound} inequality, replacing at each step the last part of the sum $\texttt{gainState}\ \sigma_{i}\ \sigma_{i+1} + \texttt{MEV\_guess}\ \sigma_{i+1}$ with its upper bound
$\texttt{MEV\_guess}\ \sigma_{i}$. At the end, the whole sum
is reduced to $\texttt{MEV\_guess}\ \sigma$.
\end{proof}

The \lstinline{MEV.characterization} theorem ensures that if there exists a  sound and coherent guess function, then that function actually provides the MEV.
We also show that our characterization is \emph{complete}: whenever MEV exists, there also exists a sound and coherent guess function (\Cref{fig:lean:MEV.characterization_complete}).

\begin{figure}[t]
\begin{lstlisting}[frame=single]
theorem MEV.characterization_complete
  (σ : sys.sysState)
  (inv : sys.sysState → Prop)
  (invariant_init : inv σ)
  (invariant_sound : 
    ∀ {m : sys.Move} {σ₀ σ₁ : sys.sysState},
    sys.semantics σ₀ m = some σ₁ →
    inv σ₀ → inv σ₁)
  (mev : ∀ σ', inv σ' → ∃ v, MEV sys σ' v) :
  ∃ guess : (σ' : sys.sysState) → inv σ' → ℝ≥0,
    -- guess is coherent
    (∃ tr : List sys.Move,
      gainMoves sys σ tr = guess σ invariant_init) ∧
    -- guess is sound
    (∀ {m : sys.Move} {σ₀ σ₁ : sys.sysState},
      (hmove : sys.semantics σ₀ m = some σ₁) →
      (σ₀_inv : inv σ₀) →
      let σ₁_inv := invariant_sound hmove σ₀_inv
      gainState sys σ₀ σ₁ + guess σ₁ σ₁_inv ≤
      guess σ₀ σ₀_inv)
\end{lstlisting}
\negcaptionspace
\caption{MEV characterization theorem (completeness).}
\label{fig:lean:MEV.characterization_complete}
\end{figure}

\begin{proof}[Proof (sketch)]
Since by assumption MEV exists, we choose as guess function the one that maps each (invariant-abiding) state to its MEV.
Coherence is straightforward, since the definition of MEV (\lstinline{trace_reaches_v}) ensures that there is a trace extracting that value.

For soundness, let \lstinline{m} be a move from state \lstinline{σ₀} to \lstinline{σ₁}. We also let \lstinline{tr} be the trace providing MEV for \lstinline{σ₁} (\lstinline{trace_reaches_v}), 
hence \lstinline{gainMoves tr = guess σ₁}.
The definition of MEV (\lstinline{other_traces_worse}) ensures that the value extracted by any trace \lstinline{tr'} from \lstinline{σ₀} is bounded by the MEV. 
In particular, choose \lstinline{tr' = m::tr}, obtaining \lstinline{gainMoves σ₀ (m::tr) ≤ guess σ₀}. 
From this and theorem \lstinline{gainMoves_step}, we conclude:
\begin{align*}
& \mbox{\lstinline{gainState σ₀ σ₁ + guess σ₁}}
\\
= \; & \mbox{\lstinline{gainState σ₀ σ₁ + gainMoves tr}}
\\
= \; & \mbox{\lstinline{gainMoves σ₀ (m::tr)}} 
\\
\leq\; & \mbox{\lstinline{guess σ₀}}
\tag*{\qedhere}
\end{align*}
\end{proof}

\zunnote{Forse si potrebbe definire un tipo "guess function" che contiene funzione, coerenza e soundness.}



Overall, formalizing our system model and the MEV-related theorems required ${\sim}$600 lines of Lean code.

\subsection{Example: MEV of the Airdrop contract}
\label{sec:mev:airdrop}

We now analyse the MEV of the Airdrop contract from~\S\ref{sec:model:airdrop}.
To this purpose, we exploit our MEV characterization theorem.
We start by defining an invariant. For this basic contract, a trivial invariant suffices.
\begin{lstlisting}
def Airdrop_invariant (σ : ADState) : Prop := True
\end{lstlisting}

We then define our guess function. As intuition suggests, we guess that the MEV is obtained by transferring the entire contract balance to the adversary.

\vbox{
\begin{lstlisting}
def Airdrop_MEV_guess (σ : ADState)
  (σ_inv : Airdrop_invariant σ) : ℝ≥0 := σ.s.bal .τ₀
\end{lstlisting}
}

Finally, by leveraging our MEV characterization theorem, we establish that our guess is indeed the MEV.

\vbox{
\begin{lstlisting}
theorem MEV_Airdrop (σ : ADState) :
  MEV Airdrop σ (Airdrop_MEV_guess σ True.intro)
\end{lstlisting}
}

\section{MEV of the CoinPusher contract}
\label{sec:coinpusher}

We now formalize the CoinPusher contract we described in~\S\ref{sec:background:coinpusher}. Some parts are similar to the Airdrop contract: \eg, types \lstinline{Participant} and \lstinline{Token} are the same, since we use the same participants and token types.
Transactions are instead modelled by the new type:

\vbox{
\begin{lstlisting}
inductive Tx
| push (P : Participant) (v : ℝ+) : Tx
\end{lstlisting}
}
A transaction \lstinline{push P v} sends \lstinline{v} tokens (of type \lstinline{τ₀}) to the contract from participant \lstinline{P}, causing \lstinline{P} to win the entire contract balance if \lstinline{v} tips such balance over the threshold.

Then, we let the adversary send only their own tokens:
\begin{lstlisting}
inductive advTx : Tx → Prop
| push {x} : advTx (.push .Adv x)
\end{lstlisting}

The contract state contains the same fields as in the Airdrop contract, plus an additional field to represent a generic \lstinline{threshold} (which we chose to be $100$ in~\S\ref{sec:background:coinpusher}).

\vbox{
\begin{lstlisting}
structure State where
  threshold : ℝ+
  bal : Token → ℝ≥0
  wal : @Wallet Token
  mempool : AssocList TxId Tx
\end{lstlisting}
}

We denote the associated system state as \lstinline{CPState}.
\begin{lstlisting}
abbrev CPState := @SysState Token State
\end{lstlisting}

The type \lstinline{Move} of adversarial moves is identical to that for the
Aidrop contract: the adversary can append to the blockchain either one of its own transactions or a transaction from the mempool.
The semantics of a \lstinline{Move} defines the contract behaviour.
An adversarial \lstinline{drop} sends tokens from the adversary to the contract, possibily making the adversary win the game. This updates \lstinline{Δ} and \lstinline{bal} accordingly.
Instead, a mempool move sends tokens from honest participants to the contract (if there are enough), possibily making the participants win the game. This updates \lstinline{wal} and \lstinline{bal} accordingly.

\vbox{
\begin{lstlisting}
def semTx (tx : Tx) (σ : CPState) : Option CPState := …
def semMove (σ : CPState) (m : Move) : Option CPState :=…
\end{lstlisting}
}

We can finally define the CoinPusher contract
\begin{lstlisting}
def CoinPusher: System := …
\end{lstlisting}

We now turn to the study of MEV. Unlike for the Aidrop contract, the MEV is now actually affected by the mempool. 
We therefore choose to establish MEV in two distinct cases:
\begin{itemize}
\item the case where the mempool is empty, and
\item the case where the mempool contains one transaction.
\end{itemize}
As we will see, the general case where the mempool contains $N$ transactions is a simple extension of these two cases: roughly, the strategy of the adversary is to apply iteratively $N$ times the strategy for the singleton mempool. 
Because of this, we focus on the first two core cases.

For the former case, we declare the following invariant.
\begin{lstlisting}
def CoinPusher_empty_invariant (σ : CPState) : Prop :=
  σ.s.mempool.isEmpty
\end{lstlisting}

If there are no transactions in the mempool, the optimal strategy for the adversary is to \lstinline{push} enough tokens to trigger the win. For the sake of simplicity, below we choose to send \lstinline{σ.s.threshold} tokens, even if a smaller amount
\linebreak
(\lstinline{σ.s.threshold - σ.s.bal .τ₀}) would also suffice.
\begin{lstlisting}
def CoinPusher_strategy_empty (σ : CPState) :
    List Move := [.adv (.push .Adv σ.s.threshold) .push]
\end{lstlisting}

We define a guess function for the MEV, posing that we can extract the entire contract balance as MEV.
\begin{lstlisting}
def CoinPusher_empty_MEV_guess (σ : CPState) 
    (σ_inv : CoinPusher_empty_invariant σ) : ℝ≥0 :=
  σ.s.bal .τ₀
\end{lstlisting}

We can prove that our guess function is sound and coherent (exploiting our strategy). This makes it possible to invoke our MEV characterization theorem and establish MEV for the empty mempool case.

\vbox{%
\begin{lstlisting}
theorem MEV_CoinPusher_empty (σ : CPState) :
  σ.s.mempool = .nil →
  MEV CoinPusher σ (CoinPusher_empty_MEV_guess σ)
\end{lstlisting}
}

We now turn to the case where the mempool contains one transaction, \ie, it is a singleton list.
We use the following invariant. Note that we have to account for the mempool becoming empty after its transaction is appended, so the invariant actually requires that the mempool is either a singleton or empty. For the sake of simplicity, we also require that if the mempool is a singleton, its transaction is from the honest participants. Indeed, adversarial transactions in the mempool are irrelevant for the purposes of MEV, so we can ignore them. 
\vbox{%
\begin{lstlisting}
def CP_one_or_less_invariant (σ : CPState) : Prop :=
 σ.s.mempool.isEmpty ∨
 (∃ idx tx, σ.s.mempool = .cons idx tx .nil ∧ tx.P ≠ .Adv)
\end{lstlisting}
}

The guess function is more complex \wrt the empty mempool case. We define it by cases, according to the mempool:
\begin{enumerate}
\item If the mempool is empty, our guess is \lstinline{σ.s.bal .τ₀}, the same we mentioned earlier.
\item Otherwise, if the mempool is the singleton \lstinline{tx} (a \lstinline{push} transaction from the honest participants), the best strategy for the adversary is 
\begin{enumerate*}[$(i)$]
    \item first, trigger a win in the contract, emptying its balance,
    \item attempt to append \lstinline{tx}, 
    \item trigger the win again.
\end{enumerate*}
The first step extracts \lstinline{σ.s.bal .τ₀} tokens, emptying the contract balance.

The second step, appending \lstinline{tx}, might fail if the honest participant does not have enough tokens to execute it, \ie if 
\lstinline{tx.v > σ.s.wal .τ₀}, in which case it has no effect.
Further, even if \lstinline{tx} succeeds, it might send to the contract enough tokens to immediately trigger the win for the honest participant, if \lstinline{tx.v ≥ σ.s.threshold}.
In this case, the tokens are immediately sent back to the participant, and the transaction has no net effect, since the contract balance was empty.
Otherwise, \lstinline{tx} succeeds and sends fewer tokens than the threshold, loading the contract balance with \lstinline{tx.v} tokens.
Overall, while the second step does not make the adversary directly gain anything, it can potentially load the contract balance.

Finally, the third step extracts the new contract balance. This could be \lstinline{tx.v} or zero depending on whether \lstinline{tx} succeeded in loading the contract or not, respectively. 

Coherently with our strategy, we guess the MEV to be \lstinline{σ.s.bal .τ₀ + tx.v} or just \lstinline{σ.s.bal .τ₀}.

\end{enumerate}

\vbox{
\begin{lstlisting}
def CoinPusher_one_or_less_MEV_guess (σ : CPState)
    (σ_inv : CP_one_or_less_invariant σ) : ℝ≥0 :=
  match σ.s.mempool with
  | .nil => σ.s.bal .τ₀
  | .cons id tx .nil =>
    if tx.v < σ.s.threshold ∧ tx.v ≤ σ.s.wal .τ₀
    then σ.s.bal .τ₀ + tx.v
    else σ.s.bal .τ₀
  | .cons id tx (.cons id2 tx2 rest) => by
    exfalso; … -- contradicts σ_inv
\end{lstlisting}
}

\vbox{%
\begin{lstlisting}
def CP_strategy_singleton_mempool (σ : CPState)
    (id : TxId) (tx : Tx) : List Move :=
  [ .adv (.push .Adv σ.s.threshold ) .push
  , Move.mempool id
  , .adv (.push .Adv σ.s.threshold ) .push ]
\end{lstlisting}
}

We can prove that our guess function is sound and coherent (exploiting our strategy). This makes it possible to invoke our MEV characterization theorem and establish MEV for the singleton mempool case.
\vbox{%
\begin{lstlisting}
theorem MEV_CoinPusher_singleton
  (σ : CPState) (id : TxId) (tx : Tx)
  (singleton : σ.s.mempool = .cons id tx .nil)
  (nonadv : tx.part ≠ .Adv) :
  MEV CoinPusher σ
    (if tx.v < σ.s.threshold ∧ tx.v ≤ σ.s.wal .τ₀
     then σ.s.bal .τ₀ + tx.v
     else σ.s.bal .τ₀)
\end{lstlisting}
}

We remark that this result extends to the general case where the mempool contains any number of transactions.
The optimal strategy here is to attempt to append all the mempool transactions, while interleaving them with an adversarial transaction that triggers the win and empties the contract balance: \lstinline{[trigger₀, tx₁, trigger₁, tx₂, …, txₙ, triggerₙ ]}.

Overall, the formalization of the CoinPusher contract and the proofs to establish its MEV amount to ${\sim}$1000 lines of Lean code.
This required more effort than the Airdrop contract, because the value extraction strategy is no longer trivial, but did not pose a significant challenge.

\section{MEV of the AMM contract}
\label{sec:amm}

We now move to our main contribution, the formalization of an Automated Market Maker contract, which we described in~\S\ref{sec:background:amm}. 

For this contract, the \lstinline{Participant} type is akin to the previous examples, with one adversary, and infinitely many honest ones.
The AMM swaps exchange two different types of tokens:

\vbox{
\begin{lstlisting}
inductive Token
| τ₀ : Token
| τ₁ : Token
\end{lstlisting}
}

The only operations supported by our AMM are tokens swaps. We define transactions accordingly.

\vbox{
\begin{lstlisting}
inductive Tx
| swap (P : Participant) (v₀ : ℝ+)
       (τ : Token) (vmin : ℝ≥0) : Tx
\end{lstlisting}
}

When a \lstinline{swap P v₀ τ vmin} is executed, participant \lstinline{P} exchanges \lstinline{v₀} of their own \lstinline{τ} tokens with at least \lstinline{vmin} tokens of the other type in the AMM reserves.
If such an exchange is impossible, \ie when \lstinline{P} does not own enough tokens or when the AMM exchange rate would cause fewer than \lstinline{vmin} tokens to be exchanged, the transaction has no effect.

The contract state features the same fields as the Airdrop contract (\ie \lstinline{bal}, \lstinline{wal}, and \lstinline{mempool}). The resulting system state is then denoted with with \lstinline{AMMState}.

Again, the type \lstinline{Move} is analogous to the other examples: the adversary can either append a transaction from the mempool or one of their own. Like in the \lstinline{CoinPusher} example, an adversarial \lstinline{swap} updates \lstinline{∆} and \lstinline{bal}, while a mempool move exchanges tokens between a honest
participant and the contract, updating \lstinline{wal} and
\lstinline{bal} accordingly.

To study the MEV of the AMM contract, we focus on two cases depending on the initial state of the mempool:
\begin{itemize}
\item the mempool is empty, and
\item the mempool contains one transaction.
\end{itemize}
As we will see, the general case where the mempool contains
$N$ transactions is a simple extension of these two cases:
roughly, the strategy of the adversary is to apply iteratively
$N$ times the strategy for the singleton mempool. Because of
this, we focus on the first two core cases.

\subsection{Empty mempool}
\label{sec:amm:empty-mempool}

For the empty mempool case, we define a simple invariant:
\begin{lstlisting}
def empty_mempool_invariant (σ : AMMState) : Prop :=
  σ.s.mempool.isEmpty
\end{lstlisting}
We easily prove the invariant to be sound.

We now define our guess function for the empty mempool case. We start by defining \lstinline{extractable σ} as the value that can be extracted by rebalancing the AMM. This can be expressed with the following formula, which we derive from~\cite{BCL22fc}.

\vbox{
\begin{lstlisting}
def extractable (σ : AMMState) : ℝ≥0 :=
  ⟨ ( Real.sqrt (pr .τ₀ * σ.s.bal .τ₀) - 
      Real.sqrt (pr .τ₁ * σ.s.bal .τ₁) )^2, ... ⟩
\end{lstlisting}
}

We then define our guess function as \lstinline{extractable σ}.

\vbox{
\begin{lstlisting}
def AMM_MEV_guess (σ : AMMState) 
    (σ_inv : empty_mempool_invariant σ) : ℝ≥0 :=
  extractable σ
\end{lstlisting}
}

We finally prove that our guess is indeed the MEV.
\vbox{
\begin{lstlisting}
theorem MEV_empty_mempool (σ : AMMState)
  (σ_inv : σ.s.mempool = .nil) :
  MEV (AMM pr) σ (AMM_MEV_guess pr σ σ_inv)
\end{lstlisting}
}
This effectively guarantees that the best adversarial strategy, when the mempool is empty, is to rebalance the AMM.
Overall, formalizing the AMM contract required ${\sim}$700 lines of Lean code, while establishing MEV in the empty mempool case required further ${\sim}$500 lines. The main challenge here was to deal with the mathematical formulas that arise in the proofs.

\subsection{Non-empty mempool}
\label{sec:amm:nonempty-mempool}

We now turn to the case where the mempool contains one transaction \lstinline{tx}.
Here, our invariant actually requires that the mempool contains only \lstinline{tx} \emph{or is empty}, since \lstinline{tx} is consumed when it is fired.
Further, we also require that \lstinline{tx} is owned by a honest participant (\lstinline{tx.part ≠ .Adv}) and that \lstinline{tx} requires a minimum positive amount of tokens in exchange (\lstinline{tx.vmin > 0}).
Indeed, adversarial transactions in the mempool are pointless, since they can be crafted, so we can rule them out. 
Further, having a mempool transaction with \lstinline{tx.vmin = 0} leads to a corner case in which MEV does not exist%
\iftoggle{anonymous}{.}{(see \Cref{app:corner-case}).}
\begin{lstlisting}
def inv_one_or_less (σ : AMMState) : Prop :=
  σ.s.mempool.isEmpty ∨
  (∃ idx tx, σ.s.mempool = .cons idx tx .nil ∧
             tx.vmin > 0 ∧ tx.part ≠ .Adv)
\end{lstlisting}

We now define our guess function, which turns out to be significantly more complex than that for the empty mempool case.
We have to provide a guess for all the states satisfying the invariant, so we must handle mempools having size zero or one.
When the mempool is empty, we guess 
\lstinline{extractable σ}, coherently with the empty mempool case.
When there is one transaction \lstinline{tx} in the mempool, we proceed as follows.

We first check if \lstinline{tx} can be beneficial to the adversary.
This happens when:
\begin{enumerate}
    \item Transaction \lstinline{tx} can actually be executed in some AMM state, \ie, its honest sender actually owns the tokens \lstinline{tx} is trying to swap (\lstinline{tx.inputVal ≤ σ.s.wal tx.inputTok}).

    \item Transaction \lstinline{tx} causes its honest sender to transfer to the contract an amount of tokens whose value (\lstinline{tx.inputVal * pr tx.inputTok}) is larger than the value of the minimum amount of tokens it is requiring in exchange (\lstinline{tx.vmin * pr (other_tok tx.inputTok)}).
\end{enumerate}

If either of the above conditions does not hold, \lstinline{tx} is useless for the adversary, who can then disregard \lstinline{tx}. In this case, our guess is the same as the one for the empty mempool case, hence \lstinline{extractable σ}.

Instead, when both of the above conditions hold, we exploit \lstinline{tx} as follows.
First, the adversary employs its unlimited token reserves and performs a swap \lstinline{move}, bringing the AMM to a state \lstinline{σ'} where executing \lstinline{tx} will transfer to the honest participant exactly \lstinline{tx.vmin} tokens.
Intuitively, from the point of view of the honest participant, \lstinline{σ'} is the \emph{worst state} where \lstinline{tx} still can be executed.
Dually, \lstinline{σ'} is also the \emph{best state} for the adversary in which to execute \lstinline{tx}.
We refer to \lstinline{σ'} as the \emph{tight state} for \lstinline{tx}.
After the \lstinline{move}, having reached the tight state the adversary indeed appends \lstinline{tx} to the blockchain. After that, we reach a new state \lstinline{σ''} where the mempool is empty, so the adversary proceeds as in the empty mempool case.

We coherently define our guess function as the sum of:
\begin{enumerate}
\item the gain of \lstinline{move} (\lstinline{gainMoves (AMM pr) σ [move]});
\item the gain of \lstinline{tx}, which is zero, since it does not directly transfer tokens to or from the adversary;
\item the gain from the rebalancing, \ie, \lstinline{extractable σ''}.
\end{enumerate}

We remark that the gain of the first step may be negative. The adversary can therefore temporarily \emph{lose} value in that step, which makes the overall MEV strategy non-trivial.
In our Lean development, we first defined our guess function as having codomain \lstinline{ℝ}. We then showed that its result is always non-negative, allowing its codomain to be restricted to \lstinline{ℝ≥0}, as required by our
\lstinline{MEV.characterization} theorem.
This actually shows how the temporary loss of value for the adversary at the first step is then compensated by the next steps.

We finally remark that, as a corner case, it is possible that the initial state \lstinline{σ} is already tight for \lstinline{tx}. 
If so, no \lstinline{move} is needed, and we simply use \lstinline{extractable σ} as our guess.

\vbox{%
\begin{lstlisting}
def AMM_MEV_guess' (σ : AMMState)
    (σ_inv : inv_one_or_less σ) : ℝ :=
  match σ.s.mempool with
  | .nil => extractable pr σ
  | .cons id tx .nil =>
    if tx.inputVal ≤ σ.s.wal tx.inputTok ∧ 
       (tx.inputVal * pr tx.inputTok > 
        tx.vmin * pr (other_tok tx.inputTok))
    then ...
      let σ' := tight_state_from_mempool σ id tx ...
      let σ'' := state_after_tight σ id tx ...
      match move_from_to_state σ σ' with
      | .some move =>
        gainMoves (AMM pr) σ [move] + extractable pr σ''
      | .none => -- σ = σ', no move needed
        extractable pr σ''
    else
      extractable pr σ
  | _ => ... -- contradicts the invariant

def AMM_MEV_guess (σ : AMMState)
    (σ_inv : inv_one_or_less σ) : ℝ≥0 :=
  ⟨ AMM_MEV_guess' pr σ σ_inv , ... ⟩
\end{lstlisting}
}

We finally establish MEV for the AMM contract when the mempool contains a single transaction.
To do so, we prove our guess to be coherent, showing that its value can actually be extracted using a trace.
We also prove our guess to be sound, showing that no adversarial move can extract more value.

\vbox{%
\begin{lstlisting}
theorem MEV_singleton_mempool (σ : AMMState)
  (id : TxId) (tx : Tx) (vmin_pos : tx.vmin > 0)
  (singleton : σ.s.mempool = .cons id tx .nil)
  (part_nonAdv : tx.part ≠ .Adv) :
  let σ_inv : inv_one_or_less σ := ...
  MEV (AMM pr) σ (AMM_MEV_guess pr σ σ_inv)
\end{lstlisting}
}

Overall, establishing MEV in the singleton mempool case required ${\sim}$3200 lines of Lean code, on top of the base AMM model (${\sim}$700 lines) and the empty mempool case (${\sim}$500 lines).
The main challenge here was dealing with the complex extraction strategy, requiring the proofs to consider all the possible cases.
Further, the involved mathematical formulas were more complex, and required more effort to handle. 

\section{Limitations}
\label{sec:limitations}

In this~\namecref{sec:limitations} we discuss how our design choices affect the   practical applicability of the proposed Lean formalization. 

\paragraph*{Real \emph{vs.} integer arithmetic}

In our MEV formalization, we represent token balances and their market values as real numbers.
In practice, however, blockchain platforms use high-precision integers to encode such amounts.
For instance, in Ethereum it is common to employ the integer type \solcode{uint256}, since floating point types are not supported.
As a consequence, when a financial contract performs its numerical computations --- \eg, calculating the exchange rate in a \lstinline{swap} of an AMM --- the result is inevitably affected by a small rounding error.

Although these errors usually have a negligible economic impact on the real-world usage of the contract,
it is possible to craft artificial scenarios where the MEV differs substantially depending on whether rounding errors are present or not.
For instance, consider a slightly modified \lstinline{CoinPusher} where the winning threshold is $9991$, and a $0.1\%$ fee is subtracted from the value sent to the contract upon each call to \lstinline{push}. 
Consider a state where the contract has zero balance and the mempool contains a \lstinline{push} transaction where a honest participant is sending $10001$ tokens.
Using real numbers, the \lstinline{push} would compute the fee as $10.001$ and transfer $9990.999$ to the contract --- so, slightly less than the threshold required to trigger the win.
The adversary could extract a MEV of $9990.999$ by back-running the mempool transaction with another \lstinline{push} of $1.0001$ or more, triggering the win.
Instead, using integers, the \lstinline{push} would round the fee to \lstinline{10}, and transfer \lstinline{9991} token units to the contract balance --- \ie exactly the threshold value.
Here, the MEV is zero, since the mempool transaction would immediately trigger the win for the honest user, hence it is useless for the adversary.

Adapting our system model and MEV formalization to use integers instead of reals is easy,
%
but we anticipate that it would significantly complicate the reasoning needed to establish MEV for certain contracts.
For instance, in our AMM, performing two consecutive adversarial \lstinline{swap}s has the same effect as a combined single \lstinline{swap}, which simplifies the treatment.
This is no longer the case when rounding errors are taken into account.
On the positive side, using integers would ensure that  MEV always exists.
This does not hold with real numbers: an \lstinline{Airdrop} variant which allows to withdraw (\lstinline{drop}) any amount \emph{smaller} than its balance has no associated MEV.
Indeed, any strategy can be improved by adding one more \lstinline{drop} to grab a tiny amount of additional tokens.

\paragraph*{Proof automation}

While Lean does provide several tactics to automate the proof of certain kinds of mathematical goals, in our experience we often wished for more powerful arithmetic simplification tactics. 
For instance, the proofs for our AMM model often make use of inequalities involving square roots.  We managed to prove these only through several manual algebraic steps, invoking each time a suitable theorem from the mathematical library. This could change in the future as Lean improves its tactics.

\paragraph*{Adversary model}

Our system model keeps the adversary syntactically distinct from honest users. 
This choice requires a certain care when encoding smart contracts into our system model, especially for contracts implementing access control mechanisms. 
In general, the system designer must ensure that the \lstinline{Move} type and its associated semantics faithfully capture all the possible adversarial actions. Omitting even a single adversarial move could lead to overlooked attacks.

\section{Related work}
\label{sec:related}

\paragraph*{Analysis tools for MEV}


The seminal work~\cite{Babel23clockwork} was the first to propose a general definition of MEV and a tool to estimate MEV upper bounds.
Their verification technique is based on a symbolic semantics that over-approximates the set of execution paths and their reachable states. 
Based upon this symbolic semantics, they encode the problem of estimating a MEV bound as a reachability problem.
Compared to our work, this technique requires less manual effort, since the designer is only required to provide a specification of the contract (in their domain-specific language), and from that point the tool automatically performs the analysis of MEV.
A main drawback of the approach is the lack of scalability (as observed in~\cite{BabelJ0KKJ23ccs}): even for relatively simple contracts, the exponential blowup of the generated paths may lead the tool to exhaust the computational resources.  
Another difference with our work is that our notion of MEV is \emph{exact}: when $\texttt{MEV} \; \texttt{sys}\; \sigma \; v$ holds, it actually means that $v$ can be extracted, and there is no adversarial strategy that can extract more.
The technique in~\cite{Babel23clockwork} instead \emph{over}-approximates MEV, without guaranteeing that the verified upper bound can actually be extracted. 


The work~\cite{BabelJ0KKJ23ccs} approached the problem of MEV estimation from a different perspective, by proposing a machine-learning technique to
\emph{under}-approximate MEV.
Their algorithm takes as input a blockchain state and a mempool, and gives as output a sequence of transactions (containing both transactions crafted by the adversary and transactions from the mempool) that can be played by the adversary to extract value.
This approach is more scalable than~\cite{Babel23clockwork}, and allows the adversary to fine-tune the computational resources used by the algorithm in order to match the available hardware.
Of course this technique does not guarantee that the synthesised sequence of transactions is optimal: instead, our approach pursues the goal of certifying the optimality of the adversarial strategy. 


\paragraph*{Formalization of the adversary}

Our definition of MEV keeps the adversary distinct from the other participants --- so being similar in spirit to~\cite{Babel23clockwork}, where MEV is parameterized by the player who extracts it.
As noted in~\S\ref{sec:limitations}, this requires some extra care when modelling the contract behaviour, to avoid the risk of omitting some adversarial actions.
Some approaches to avoid fixing the identity of the adversary have been proposed in literature.
The work~\cite{Salles21formalization} defines an adversary-agnostic MEV by first considering the MEVs that can be extracted by any individual participant, and then taking the minimum.
As noted in~\cite{Salles21formalization}, this definition has some drawbacks when extracting value requires the adversary to make upfront payments. 
In such cases, the MEV is under-estimated as zero, since the minimum also covers the zero value that can be extracted by a ``poor'' adversary.
The notion of \emph{universal MEV} proposed in~\cite{BZ25fc} overcomes this problem by defining MEV as a game where honest players try
to minimize the damage, while adversaries try to maximize their gain.
A token redistribution mechanism ensures that the adversaries have enough wealth to execute their extraction strategy. 
The token redistribution is not completely equivalent to our wealthiness assumption about adversaries (\S\ref{sec:model:systate}), since the token to be distributed may not be sufficient for the attack. 
A wealthiness assumption more similar to ours was used in~\cite{BMZ24fc}, which introduces the notion of ``MEV of wealthy adversaries'' by taking the maximum MEV over all possible adversarial wallets.

\paragraph*{Analysis of the AMM contract}

Compared to real-world AMM implementations such as Uniswap~\cite{uniswapimpl}, our Lean formalization introduces a few simplifications, that overall contribute to keeping our proofs manageable.
%
A first simplification is that in Uniswap, when a user executes a \lstinline{swap}, part of the input tokens are paid to the protocol as a fee.
Studying the effect of swap fees on MEV would require a further complication of the adversarial strategies, who would need to minimize the impact of fees on their own swaps.
The Lean formalization in~\cite{Dessalvi25github} studies AMMs with swap fees, but unlike ours, it does not address MEV.
The additional burden introduced by swap fees is, however, already visible in~\cite{Dessalvi25github} in the analysis of arbitrage (\ie, the optimal swap in a zero-player game).
We expect swap fees to impact MEV analysis to an even greater extent.

Besides swaps, concrete AMM implementations typically allow users to deposit and withdraw tokens from the AMM.
In particular, deposit transactions can increase the profits obtainable from subsequent swaps~\cite{BCL22amm}, and can thus be leveraged by an adversary to amplifly MEV.
The work~\cite{BCL22fc} introduces an extended sandwich attack that combines deposit and swap transactions from the mempool with adversary-crafted transactions, showing that this may increase the extractable value.
Extending our proof to handle mempools containing arbitrary combinations of deposit and swap transactions would, however, cause a substantial explosion in the number of cases in our proof of MEV optimality.

A different Lean formalization of constant-product AMMs is presented in~\cite{Pusceddu24fmbc}.
Compared to our work --- which proposes a general framework for analysing MEV of arbitrary smart contracts --- the work~\cite{Pusceddu24fmbc} is specifically targeted on AMMs, and only deals with arbitrage.

Sandwich attacks on AMMs were originally analysed in~\cite{Zhou21high}, and later in~\cite{BCL22fc,Xu23csur,Wang23esorics,Park2023mansci,Kulkarni23wine}, among the others.
Compared to these works, ours provides the first machine-checked proof of optimality of sandwich attacks. 


\paragraph*{MEV countermeasures}

Given the relevance of MEV on the blockchain ecosystem, several techniques to mitigate its effects have been proposed in the past few years. 
Some of these techniques are applicable to arbitrary smart contracts~\cite{breidenbach2017hydra,Baum21iacr,Heimbach22aft,Canidio24mansci,Babel24arxiv}, while some others are specific to certain classes of contracts, such as AMMs~\cite{Ciampi22cscml,baump2dex}.
Formally verifying within our Lean framework that some of these techniques achieves the expected MEV reduction would be a challenging extension of our work.



\section{Conclusions}
\label{sec:conclusions}

We present a Lean framework for reasoning about MEV in smart contracts. The framework is blockchain-agnostic and can be instantiated to model different smart contract languages, making it applicable to a wide range of real-world scenarios.
Unlike prior tools that are focussed on under- or over- approximations, our proof technique can establish \emph{exact} MEV bounds, by identifying attacks that are realistically executable by an adversary and proving their optimality (no strictly more profitable attacks exist). 
We validate the framework by formalising representative DeFi protocols and mechanising their analysis in Lean.
In particular, we provide the first machine-checked proof of the optimality of sandwich attacks on Uniswap v2-based Automated Market Makers, demonstrating the effectiveness of our approach.

\iftoggle{anonymous}{}{\section*{Acknowledgment}

Work partially supported by project SERICS (PE00000014)
under the MUR National Recovery and Resilience Plan funded by the
European Union -- NextGenerationEU, and by PRIN 2022 PNRR project DeLiCE (F53D23009130001).}

\bibliography{main}
\bibliographystyle{IEEEtran}

\iftoggle{anonymous}
{}
{
    \appendix

\subsection{Non existence of MEV in the general case}
\label{app:corner-case}

In this appendix we consider the corner case in which MEV fails to exist for an AMM contract. 

Specifically, we consider an AMM contract, starting in a state \lstinline{σ} such that the mempool in \lstinline{σ} contains a single transaction \lstinline{tx} with \lstinline{tx.vmin = 0} and \lstinline{tx.inputValue = τ₀}.
In such a scenario we can prove that \lstinline{MEV σ} does not exist, and that \lstinline{MEVSup σ} is equal to 
\lstinline{extractable σ + tx.v * pr₀}.

We prove the existence of \lstinline{MEVSup σ} with two propositions: the first one establishes that the gain of the adversary is strictly bounded from above, while the second one establishes that the adversary can extract a value that is arbitrarily close to that upper bound.

\begin{prop}
The value that the adversary can gain by performing a sequence of transactions in \lstinline{σ} is strictly less than \lstinline{extractable σ + tx.v * pr₀}. In short, we have
\vbox{
\begin{lstlisting}
  ∀ tr, gainMoves σ tr < extractable σ + tx.v * pr₀
\end{lstlisting}
}
\end{prop}
\begin{proof}
    Since the mempool contains a single transaction, we can assume that \lstinline{tr} either includes a single mempool move or none at all. Indeed, although theoretically \lstinline{tr} could include more than one mempool move, at most one of them can be executed succesfully, since the contract semantics would remove the transaction from the mempool after executing it. An unsuccessful move leaves the contract state unaltered, so we can prune these failing moves from \lstinline{tr}.

    The case in which \lstinline{tr} features no mempool transactions is simple: indeed this is akin to working in a contract with an empty mempool. 
    From our results on AMMs with empty mempool (\lstinline{MEV_empty_mempool}, \lstinline{AMM_MEV_guess} and \lstinline{exact_gain_simple}) we know that such a contract would have MEV equal to \lstinline{extractable σ}, and since both \lstinline{tx.v} and \lstinline{pr₀} are positive, we have:

    \vbox{
    \begin{lstlisting}
        gainMoves σ tr ≤ extractable σ
                        < extractable σ + tx.v * pr₀\end{lstlisting}}
    which proves our proposition.

    \smallskip
    Instead, if \lstinline{tr} contains the mempool move that executes \lstinline{tx}, we can write it as follows:
    
    \vbox{
    \begin{lstlisting}
        tr = tr' :: mempool id :: tr''\end{lstlisting}}
    where \lstinline{tr'} and \lstinline{tr''} are lists of adversarial moves, and \lstinline{id} is the transaction id associated to \lstinline{tx}.
    Thus we would have
    \begin{align}\label{eqn:appendix:gain_trace}
    \begin{split}
        \texttt{gainMoves } \sigma \texttt{ tr} =&  \texttt{ gainMoves } \sigma \texttt{ tr'}
        \\
                & + \texttt{gainMove } \sigma' \texttt{ (mempool id)} 
        \\
                & + \texttt{gainMoves } \sigma'' \texttt{ tr''} 
    \end{split}
    \end{align}
    where $\sigma'$ and $ \sigma''$ are such that $\xrightarrow{\texttt{tr'}} \sigma'$ and $\sigma' \xrightarrow{\texttt{tx}} \sigma''$.
    
    Now, let $b_0$ (respectively $b'_0$) be equal to the balance of $\sigma$ (respectively $\sigma'$) in $\tau_0$ and $b_1, b'_1$ be the balance of those contract states in $\tau_1$. 
    From this point onward we will also denote $\texttt{tx.v}$ as $v$. Finally $\texttt{pr}_0$ and $\texttt{pr}_1$ are the prices of tokens. 
    The adversarial gain can be calculated from the contract balances as follows:
    \[
        \texttt{gainMoves} \ \sigma \ \texttt{tr'} = \texttt{pr}_0 (b_0 - b'_0) +\texttt{pr}_1 (b_1 - b'_1)
    \]
    Since $\sigma' \xrightarrow{\texttt{tx}} \sigma''$ we can calculate the balance of $\sigma''$. In $\tau_0$ this balance is equal to $b'_0 + v$, and in $\tau_1$ it is $\dfrac{b'_0  b'_1}{b'_0 + v}$.
    
    To give a bound on \lstinline{gainMoves σ'' tr''} we notice that $\texttt{tr''}$ only contains adversarial moves. Therefore we can use the results on the MEV of AMMs with an empty mempool to conclude that
    \begin{align*}
        \texttt{gainMoves} \ & \sigma''  \ \texttt{tr''} \leq 
        \texttt{extractable} \  \sigma'' =
        \\ &=  \texttt{pr}_0 (b'_0 + v) + \texttt{pr}_1 \dfrac{b'_0  b'_1}{b'_0 + v}  -  2 \sqrt{ b'_0  b'_1  \texttt{pr}_0  \texttt{pr}_1} 
    \end{align*}

    Finally remembering that  \lstinline{tx} itself gains no token for the adversary, and that $b_0 b_1 = b'_0 b'_1$, we can use \Cref{eqn:appendix:gain_trace} to obtain that $\texttt{gainMoves} \ \sigma \ \texttt{tr}$
    must be smaller than
    \[
        \texttt{pr}_0 (b_0 + v) + \texttt{pr}_1 \dfrac{b'_0  b'_1}{b'_0 + v} +\texttt{pr}_1 (b_1 - b'_1)  -  2 \sqrt{ b_0  b_1  \texttt{pr}_0  \texttt{pr}_1} 
    \]
    which in turn can be rewritten as
    \[
         \texttt{extractable } \sigma + \texttt{pr}_0 v + \texttt{pr}_1 \dfrac{b'_0  b'_1}{b'_0 + v} - \texttt{pr}_1 b'_1  
    \]
    Since $\dfrac{b'_0}{b'_0 + v} < 1$ we have shown that every trace extracts a value that is below $\texttt{extractable } \sigma + \texttt{pr}_0 v$.
\end{proof}

\begin{prop}
There exists a sequence moves \lstinline{tr} that the adversary can perform in \lstinline{σ} order to gain an amount of value that is arbitrarily close to \lstinline{extractable σ + tx.v * pr₀}. In short, we have:
\end{prop}
\vbox{
\begin{lstlisting}
   ∀ (ε : ℝ+), ∃ tr, 
   gainMoves σ tr + ε ≥ extractable σ + tx.v * pr₀
\end{lstlisting}
}

\begin{proof}
To prove this proposition we will calculate the gain of a trace parametrized by the positive real value $x \in \mathbb{R}^+$.

First, we name some contract states and transactions.
\begin{itemize}
    \item  $\overline{\sigma}$ is the state in which the AMM is balanced, \ie the state in which 
    \[
        \overline{b}_0 \cdot \texttt{pr}_0 = \overline{b}_1  \cdot \texttt{pr}_1
    \]
    where $\overline{b}_0$ (respectively $\overline{b}_1$) is the balance of $\overline{\sigma}$ in $\tau_0$ (respectively $\tau_1$).

    \item $m_\sigma$ is the transaction that balances the AMM when applied in $\sigma$, \ie the transaction such that
    \[
        \sigma \xrightarrow{m_\sigma} \overline{\sigma}
    \]
    \item $\texttt{adv}(x)$ is an adversarial transaction that sends $x$ tokens of type $\tau_0$ to the contract.
    \item $\sigma'$ and $\sigma''$ are states such that
    \[
        \overline{\sigma} 
        \xrightarrow{\texttt{adv}(x)} \sigma' \xrightarrow{\texttt{tx}} \sigma''
    \]
    It easily follows from the definition of $\texttt{adv}(tx)$ that the balance of $\sigma'$ in $\tau_0$ is $\overline{b}_0 + x$, and that its balance in $\tau_1$ is $ \frac{\overline{b}_0 \overline{b}_1}{\overline{b}_0 + x} $.
    Similarly, from the definition of  $\texttt{tx}$, we have that the balance of $\sigma''$ is $\overline{b}_0 + x + \texttt{tx}.v$ in $\tau_0$ and $\frac{\overline{b}_0 \overline{b}_1}{\overline{b}_0 + x + \texttt{tx}.v} $ in $\tau_1$. 
    
    \item $m_{\sigma}$ is the transaction that balances the AMM when applied in $\sigma''$
\end{itemize}

We define  $\texttt{tr}(x)$ as the following list of moves: 
\[
   \texttt{tr}(x) := m_\sigma  :: adv(x) :: \texttt{tx} :: m_{\sigma''} 
\]

We will now calculate the amount of tokens that the adversary gains from appending $\texttt{tr}(x)$ to the blockchain.
\begin{itemize}
    \item $m_\sigma$ yields
    \[
        \texttt{pr}_0 (b_0 - \overline{b}_0) + \texttt{pr}_1 (b_1 - \overline{b}_1) 
    \]
    \item $\texttt{adv}(x)$ yields $\texttt{pr}_1 \left(\overline{b}_1 - \frac{\overline{b}_0 \overline{b}_1}{\overline{b}_0 + x} \right) - \texttt{pr}_0 x$ (a negative quantity).
    \item $\texttt{tx}$ is a mempool transaction, so it doesn't yield any gain for the adversary.
    \item Finally, $m_{\sigma''}$ yields $\texttt{extractable } \sigma''$ for the adversary which is equal to
    \[
        \texttt{pr}_0 (\overline{b}_0 + x + v) + \texttt{pr}_1 \dfrac{\overline{b}_0 \overline{b}_1}{\overline{b}_0 + x+ v} - 2\sqrt{\overline{b}_0 \overline{b}_1 \texttt{pr}_0 \texttt{pr}_1 } 
    \]

\end{itemize}
The gain of $tr(x)$, which is the sum of all the above gains, amounts to:
\begin{align*}
  & \texttt{pr}_0 (b_0 - \overline{b}_0) + \texttt{pr}_1 (b_1 - \overline{b}_1)  
  \\
  + \;  & \texttt{pr}_1(\overline{b}_1 - \dfrac{\overline{b}_0\overline{b}_1}{\overline{b}_0 + x} ) - \texttt{pr}_0 x
  \\
  + \; & \texttt{pr}_0 (\overline{b}_0 + x + v) + \texttt{pr}_1 \dfrac{\overline{b}_0 \overline{b}_1}{\overline{b}_0 + x+ v} - 2\sqrt{\overline{b}_0 \overline{b}_1 \texttt{pr}_0 \texttt{pr}_1 } 
\end{align*}
Since we are working with a constant product AMM, we have that $b_0 b_1 = \overline{b}_0 \overline{b}_1$, meaning that expression above can be simplified as 
\begin{align*}
    \texttt{pr}_0 b_0 
    \; + \; & \texttt{pr}_1 b_1 - 2\sqrt{b_0b_1 \texttt{pr}_0 \texttt{pr}_1}  + \texttt{pr}_0 v 
    \\ 
    + \; & \texttt{pr}_1 \overline{b}_0 \overline{b}_1 
    \left( \dfrac{1}{\overline{b}_0 + x + v} -  \dfrac{1}{\overline{b}_0 + x} \right)
\end{align*}
which is equal to 
\begin{gather*}
    \texttt{extractable } \sigma + \texttt{pr}_0 v
    - \texttt{pr}_1 \overline{b}_0 \overline{b}_1 v  \dfrac{1}{(\overline{b}_0 + x + v)(\overline{b}_0 + x)}
\end{gather*}
which gets arbitrarily close to $\texttt{extractable } \sigma + \texttt{pr}_0 v$ for large values of $x$.
Therefore, for all $\epsilon$ one can find a value of $x$ such that
\begin{lstlisting}
       gainMoves σ tr(x) + ε ≥ extractable σ + tx.v * pr₀
\end{lstlisting}
proving our proposition.
\end{proof}

}

\end{document}